\documentclass[12pt]{iopart}
\usepackage{bm}
\usepackage{array}
\usepackage{graphicx}
\usepackage{color}

\DeclareMathAlphabet{\mathdj}{U}{msb}{m}{n}  
\newcommand{\Reals}{\ensuremath{\mathdj {R}}} 

\newbox\squ  

\setbox\squ=\hbox{\vrule width.4pt 
   \vbox{\hrule height.4pt width.6em\kern1.4ex 
         \hrule height.4pt}%
   \vrule width.4pt depth0pt }
\def\sqbox{\copy\squ\hskip -.4pt}

\newenvironment{proof}{\topsep=\smallskipamount \partopsep=0pt  %
 \begin{trivlist} \itemindent=\parindent                        %
 \item[\hskip \labelsep\emph{Proof:}]}{\qed\end{trivlist}}      %
\let\qed=\relax                                                 %
\def\qed                                                        %
 {{\unskip\nobreak\hfil\penalty50                               %
   \quad\hbox{}\nobreak\hfil $\sqbox$                             %
   \parfillskip=0pt \finalhyphendemerits=0 \par}}               %

\newtheorem{prop}{Proposition}

\begin{document}

\title[The Conjecture of Chaotic Cosmological Inhomogeneity]{Dynamical Equilibrium States of a Class of Irrotational
	Non-Orthogonally Transitive $G_{2}$ Cosmologies I:
	The Conjecture of Chaotic Cosmological Inhomogeneity}

\author{C. G. Hewitt}

\address{$^1$ University of Waterloo, C.E.M.C., Faculty of Mathematics, Ontario, Canada}
 
\ead{cghewitt@uwaterloo.ca}
\vspace{10pt}
\begin{indented}
\item[]March 2021
\end{indented}

\begin{abstract}

The Einstein field equations for a class of irrotational non-orthogonally
transitive $G_{2}$ cosmologies are written down as a system of partial
differential equations. The equilibrium points are self-similar and can be
written as a one-parameter, five-dimensional, ordinary differential equation.
The corresponding cosmological models both evolve and have one-dimension of
inhomogeneity. The major mathematical features of this ordinary differential
equation are derived, and a cosmological interpretation is given. The
relationship to the exceptional Bianchi models is explained and exploited to
provide a conjecture about future generalizations.
\end{abstract}

%
%
%
%
%

\section{Introduction:}

In order to obtain cosmological models that have spatial structure, one is
forced to examine inhomogeneous cosmologies. The simplest evolving
inhomogeneous cosmologies admit two spacelike isometries. In a previous work,
Hewitt and Wainwright \cite{hewitt1990orthogonally}, we defined the $G_{2}$ cosmologies to be
the spacetimes admitting two commuting spacelike Killing vector fields (KVFs)
and satisfying the Einstein field equations (EFEs) with a perfect fluid
source. The EFEs for that subclass in which the $G_{2}$ acts orthogonally
transitively (OT), were written as a first order system of quasilinear partial
differential equations (PDEs). The equilibrium points of that system were
defined to be those models whose \textit{dimensionless} \textit{variables} do
not evolve, and were labelled "dynamical equilibrium states," (DES). The
corresponding cosmological models are self-similar, admitting an $H_{3}$
acting on the timelike three spaces containing the orbits of the $G_{2}$ and
the fluid flow vector. The Einstein field equations for these models reduce to
a system of ordinary differential equations (ODEs), which describes the
spatial structure of these models. A qualitative analysis of the dynamical
equilibrium states of the orthogonally transitive $G_{2}$ cosmologies was
given in \cite{hewitt1997dynamical}, and their asymptotic spatial structure was determined.
The DES of the OT $G_{2}$ cosmologies are interesting cosmological models in
their own right: open sets of these models are well behaved, having a big-bang
initial singularity and they are asymptotically spatially homogeneous at large
spatial distance, with their inhomogeneity only being significant over a
finite spatial region. In addition, it was conjectured in \cite{hewitt1990orthogonally} that these DES
may act as asymptotic states for less symmetric cosmological models at either
early or late times, or both, and explicit examples were provided with the
exact solutions of \cite{feinstein1989new}, \cite{senovilla1990new}, their
generalisations \cite{hewitt1990orthogonally}, and \cite{wainwright1980some}.
The main goal of the current paper is to extend this approach to the $G_{2}$
cosmologies in which the Abelian $G_{2}$ does not necessarily act orthogonally
transitively, and for which the fluid flow vector is orthogonal to the group orbits. We refer to
these models as the irrotational N-OT $G_{2}$ cosmologies. We complete this
goal in three major stages: first of all we write down the EFEs for the
irrotational N-OT $G_{2}$ cosmologies, then we write down the EFEs for the DES
of the irrotational N-OT $G_{2}$ cosmologies, and in the final stage we
provide a qualitative analysis of the resulting ODE and make some conclusions
about these cosmological models. A secondary goal is to provide comparison
between the structure and features of both the differential equations and
the corresponding cosmological models for the DES of the irrotational N-OT
$G_{2}$ cosmologies and the exceptional orthogonal spatially homogeneous (OSH) cosmologies. The third goal is to extend these comparisons further by introducing the conjecture of chaotic spatial inhomogeneity. In the companion paper, (Rashidi, Hewitt and Charbonneau \cite{work-in-progress}), we provide a complete mathematical analysis of the DES of the subclass in which at least one of the KVFs is hypersurface orthogonal (HO).\\
We assume throughout that:
\begin{itemize}
	\item H1: spacetime admits a two dimensional Abelian isometry subgroup,
	$G_{2}$, and its orbits are spacelike.
		
	\item H2: the source is a perfect fluid with a linear equation of state,
	$p=(\gamma-1)\mu$. Here $p$ is the pressure, $\mu$ is the energy density and
	$\gamma$ is the equation of state parameter ($1<\gamma<2)$.
\end{itemize}
These are the $G_{2}$ cosmologies.
\begin{itemize}
	\item H3: the isometry subgroup, $G_{2}$, does not (necessarily) act
	orthogonally transitively.
\end{itemize}

These are the N-OT $G_{2}$ cosmologies. Just for clarification, note that OT $G_2$ cosmologies are also N-OT $G_2$ cosmologies.\\
The fluid 4-velocity for the OT $G_{2}$ cosmologies is \textit{forced} to be orthogonal to the $G_{2}$ orbits. If we do not restrict the action of the isometry group then the fluid 4-velocity does not have any further
constraints. For simplicity we assume that:\smallskip
\begin{itemize}
	\item H4: the fluid velocity, $\mathbf{u}$ is orthogonal to the orbits of the isometry subgroup.
\end{itemize}
An immediate consequence of this restriction is that the fluid is forced to be irrotational. These are the irrotational N-OT $G_{2}$ cosmologies.

The \textbf{first stage} of the main goal of this paper is completed in the appendix by writing the EFEs for the irrotational N-OT $G_{2}$ cosmologies as a first order system of quasilinear partial differential equations (PDEs) in terms of dimensionless dynamical variables. Dynamical equilibrium states are also defined there, and it is explained that the corresponding cosmological models are self-similar, admitting a three parameter similarity group, $H_{3}$, acting on the three spaces containing the $G_{2}$ orbits and $\mathbf{u}$.\footnote{A three-parameter similarity group must admit a two-parameter isometry subgroup and in all but one case (see \cite{eardley1974commun}) this subgroup is Abelian.}. Consequently, in the main part of the paper (i.e. all except the
appendix), we assume that:
\begin{itemize}
	\item H5: spacetime admits a three-dimensional similarity group ($H_{3}$) having
	timelike orbits, and 
	
	\item H6: the fluid velocity $\mathbf{u}$ is tangential to the orbits of the similarity group.
\end{itemize}

In section 2 we provide further motivation for studying cosmological models which adhere to H1-H6. We introduce a subclassification scheme which shows the connection of the DES of the irrotational N-OT $G_{2}$ cosmologies to previously studied self-similar and spatially homogeneous cosmologies. This both partially addresses the second goal and leads naturally to the completion of the third goal.\\
The EFEs for the DES of the OT $G_{2}$ cosmologies may be written as a
\textbf{4-dimensional} system of ODEs with \textbf{two parameters}. One of the variables is a shear variable and one of the parameters is a\ \textit{shear parameter}. In section 3, we complete the \textbf{second stage }of our primary goal, by writing the EFEs for the DES\ of the irrotational N-OT $G_{2}$ cosmologies as an autonomous system of ODEs, with no constraints. The particular system is \textbf{5-dimensional} and has \textbf{one parameter}: the above-mentioned \textit{shear parameter} is fixed to a specific value, and, in addition, the models allow an additional (off-diagonal) shear variable.

The symmetries and invariant sets of the ODE are also listed in section 3. In section 4, we analyse the local stability of the equilibrium points; outline the main features of this ODE; deduce some of the properties of the solutions; make statements about open sets of trajectories, and complete the third stage of our primary goal by giving some cosmological interpretations. In section 5, we conclude by completing our second goal of discussing the analogies between the analyses of the exceptional OSH\ cosmologies and the\ DES of the N-OT irrotational $G_{2}$ cosmologies.\\
It is assumed that the reader is familiar with the orthonormal frame formalism of Ellis and MacCallum \cite{Cargesemaccallum1973cosmological}. The reader is referred to the book \cite{wainwright2005dynamicalSystems} for a thorough treatment of the techniques and methods of dynamical systems, especially when applied to cosmology. We use geometrised units with $c=1$, $8\pi G=1$, and the sign conventions of \cite{Cargesemaccallum1973cosmological}. Throughout this paper lower case indices $a,b,c$ assume the values $0,1,2,3$ while upper case indices $A,B,C$ assume the values $2,3$. \\

\section{Motivation:}

In this section we provide reasons why cosmological models
satisfying H1-H6 are worth examination. First of all,

\begin{itemize}
	\item Many of these models are well-behaved inhomogeneous cosmological models themselves.
	
	\item Inhomogeneous self-similar models have been shown to act as asymptotic
	states, at both early and/or late times, for less symmetric (non-self-similar)
	cosmological models. Thus the spatial structure of these models may provide
	insight into the possible spatial structure of generic cosmological solutions
	of the EFEs.
	
	\item These models generalise well-known spatially homogeneous models which
	appear as equilibrium points in phase space. These are some plane wave models;
	the Collins $VI_{h}$ models, the Wainwright $VI_{\frac{-1}{9}}$ model and the
	Robinson-Trautman solution.
	
	\item The additional variables which arise when moving from the OT models to
	the N-OT models are shear variables. This work can be interpreted as an
	examination of the effects that additional shear has on the
	inhomogeneity of cosmological models.
\end{itemize}

Secondly, it is useful to be able to visualise where the models considered in
this paper lie relative to previously considered models. In order to make this
comparison, it is instructive to recall, and utilise, the paper of \cite{wainwright1981exact}. In that work, spacetimes that admit an Abelian $G_{2}$ were
classified according to the nature of the action of the $G_{2}$. In addition
to whether or not the Abelian $G_{2}$ acts OT, this classification considers
whether or not there are HO KVFs. In order to clearly illustrate how the
models considered in this paper are related to previously considered
inhomogeneous models, we superimpose the Wainwright classification onto
cosmologies which satisfy restrictions H1, H2 and H4-H6 in Table \ref{TableofWainwrightSub}.
\begin{table}[h!]
\caption{Wainwright Subclassification for models satisfying H1, H2 and H4-H6.}
\label{TableofWainwrightSub}
\begin{indented}
\item[]\begin{tabular}{@{}*{5}{l}}
		\br
		{Class} & {Geometric Restriction} & Variables &
		{Parameters} & { Reference}\\ \hline
		{A(i)} & {not OT no HO KVFs} & 5 & 1 &
		{Sections 3,4}\\ \hline
		{A(ii)} & { not OT 1 HO KVF} & 3 & 1 &
		\cite{work-in-progress}\\ \hline
		{B(i)} & {OT no HO KVFs} & 4 & 2 & \cite{hewitt1997dynamical}\\ \hline
		{B(ii)} & {OT 2 orthogonal HO KVFs} & 2 & 3 &
		{\cite{hewitt1988qualitativeI} and \cite{hewitt1991qualitativeII}}\\ \hline
\end{tabular}
\end{indented}
\end{table}
\begin{itemize}
	\item It is of interest to complete the examination of the inhomogeneous cosmological models of Table \ref{TableofWainwrightSub}.
	
	\item Most studies on inhomogeneous cosmology consider models which admit an OT $G_{2}$ (see, for example the book, \cite{krasinski1996overview}). The current work considers a large class of models in which there is a N-OT $G_{2}$, and thus progress is being made towards the possibility of understanding the most general cosmological model, i.e. one which admits no symmetries
\end{itemize}
Thirdly, the models considered in this paper are very closely related to
(some) OSH\ models. In fact, the EFEs for both classes can be obtained by applying additional restrictions to the PDE in the appendix for the EFEs of the irrotational N-OT $G_{2}$ cosmologies. In order to make these similarities explicit, we have found it useful to produce another table for the OSH models. We first need to refine the hypotheses H5 and H6 by, H$^{\ast}$5 and H$^{\ast}$6:
\begin{itemize}
	\item H$^{\ast}$5: spacetime admits a three-parameter isometry group ($G_{3}$) having spacelike orbits, and
	
	\item H$^{\ast}$6: the fluid velocity $\mathbf{u}$ is orthogonal to the orbits of the
	$G_{3}$.
	
	\item On comparing these two tables, it is clear that the subdivision into the various subclasses is identical.
\end{itemize}
\begin{table}[h!]
\caption{Wainwright Subclassification for models satisfying H1, H2 and H4, H*5 and H*6.}
\label{TableofWainwrightSub0}
\begin{indented}
	\item[]\begin{tabular}{@{}*{5}{l}}
		\br
		{Class} & {Geometric Restriction} & Variables &
		{Parameters} & { Reference}\\ \hline
		{A(i)} & {not OT no HO KVFs} & 5 & 1 &
		{\cite{hewitt2003asymptotic}}\\ \hline
		{A(ii)} & { not OT 1 HO KVF} & 3 & 1 &
		 \cite{hewitt1991algebraic} \\ \hline
		{B(i)} & {OT no HO KVFs} & 4 & 2 & \cite{uggla1989asymptotic},\cite{hewitt1997bianchi}
		\\ \hline
		{B(ii)} & {OT 2 orthogonal HO KVFs} & 2 & 3 &
		 \cite{collins1971more}\\ \hline
	\end{tabular}
\end{indented}
\end{table}
In the conclusion we emphasise the similarities and differences between the N-OT subclasses in the self-similar and spatially homogeneous cases. Examination of this class of models provides a first step in the
consideration of the Chaotic Inhomogeneity Conjecture, which we now outline.

In \cite{hewitt1990orthogonally} and \cite{hewitt1997dynamical} the role played by the \textit{spatially homogeneous}
plane wave solutions as asymptotic states for \textit{spatially inhomogeneous
}models at large \textit{distance has been revealed}. There are several
analogies with the role played by the Kasner models as asymptotic states at
\textit{early times} for some OSH models (i.e. those which are also OT $G_{2}$
cosmologies). However, it is also well known that for the most general OSH
models and generic tilted spatially homogeneous models (i.e. those containing
the general tilted Bianchi II model as a subclass) the extra variables
involved have the effect of destabilising the Kasner equilibrium points and
turning them into saddle points. The dynamical effect of this is rather
drastic: a typical trajectory is no longer past asymptotic to \textit{a single
	equilibrium point}, but bounces around the phase space as it is directed from
one saddle point to another. The corresponding cosmological model has an
oscillatory initial singularity, and the word "chaos" is
often associated with the dynamical behaviour. One of the goals of this paper
is to consider whether the analogy between Kasner models and plane wave models
extends further, that is we consider the conjecture:\\ \\

\begin{center}
	\textbf{The Conjecture of Chaotic Cosmological Inhomogeneity.\footnote{This
			conjecture was briefly discussed at both the 13$^{th}$ Canadian
			Conference on General Relativity and the 12$^{th}$ Marcel Grossmann Meeting \cite{hewitt2012conjecture} }  }
\end{center}

\textit{The spatial structure of a generic spatially inhomogeneous
	cosmological model is chaotic. Generic spatially inhomogeneous cosmological
	models possess an oscillatory nature: as one moves in any given spatial
	direction the model is vacuum dominated and may be approximated by a plane
	wave model over a finite volume and then further on it may be approximated by
	another plane wave model. This sequence extents indefinitely, and between any
	two plane wave regimes there is smaller matter dominated region. The
	differential equation which determines the sequence of oscillations between
	the intermediate plane waves states is chaotic.}

The observational consequence for inhomogeneous cosmological models which
satisfy this conjecture is rather simple: a typical cosmological model is
approximately vacuum (i.e. very low energy density) almost everywhere, but not
the same type of vacuum solution. Between different vacuum regions there is a
comparatively small region where the energy density is not insignificant, and
matter close to the boundary between two regions would be accelerated towards
the smaller matter dominated region. On a large scale, these predictions are
similar to actual observations of the large scale structure of the universe as
voids (lower energy density regions) and filaments (smaller regions of higher
energy density), see for example \cite{cautun2014evolution} et al.

The research which has been carried out to produce this paper can be
considered as a first attempt to verify the Chaotic Inhomogeneity Conjecture.
On considering the analogy from the spatially homogeneous models, we expect
that it is necessary to examine inhomogeneous cosmological models in which
any $G_{2}$ does not act orthogonally transitively. Models which satisfy
assumptions H1-H6 provide a relatively simple and accessible class, however
their analysis is complicated as the phase space is non-compact.
We show that the effect of the additional shear variable is indeed to
destabilise most points on the curve of the plane-wave models, however there still
remains segments which act as a source or as a sink. We conclude that an
examination of an even more general class is required; in that case, there is
the possibility that the additional variable(s) would further destabilise the
plane wave arcs and change the stability of equilibrium points to saddles. We remark that Coley and Hervik have shown in \cite{coley2005dynamical} that for a specific class of plane wave models, each plane wave equilibrium point possesses an unstable manifold of dimension at least one, when the phase space is large enough.

\section{ EFEs for the DES of the Irrotational N-OT $G_{2}$ Cosmologies}

The dimensionless form of the EFEs for the DES of the irrotational N-OT
$G_{2}$ cosmologies have been developed in the appendix. There is rotational
freedom in the group invariant orthonormal frame and there are constaints on
the system of ODEs. When examining the OT models, \cite{hewitt1990orthogonally}, we found it convenient to
use group-invariant variables, that is variables that are not affected by the
rotational freedom in the orthonormal frame. The same approach is not suitable
for the N-OT models due to the large number of additional constraint equations
which are introduced. Instead, we choose to make use of the rotational freedom
of the orthonormal frame to set the dimensionless shear variable $\Sigma_{12}$
to zero. The same approach has been employed in other studies, including the
examination of the exceptional orthogonal Bianchi models in \cite{hewitt2003asymptotic}.

Once $\Sigma_{12}$ has been set to zero, it then follows from restriction (\ref{A.36b}) that $R=3 \tilde{\Sigma}_{23}$; from restriction (\ref{A.36a}), and the definition of
$k$ in (\ref{A.34}), that $\tilde{\Sigma}_{22}=\frac{1}{4}(6-5\gamma)$, and from
equation (\ref{A.30f}), for the spatial derivative of $\Sigma_{12}$, that
$N_{+}=3\tilde{N}_{22}(\iff N_{33}=0)$. The remaining independent non-zero
variables are the three spatial curvature variables:\ $A,\tilde{N}_{23}
,\tilde{N}_{22}$, and the two shear variables: $\tilde{\Sigma}_{23}$ and
$\Sigma_{13}$. The acceleration and the energy density may be expressed, as
quadratic and cubic polynomials respectively, in terms of the other variables.\\
The EFEs reduce to: \textit{Spatial structure equations:}
\begin{eqnarray}
\eqalign{
\boldsymbol \partial_{1}(A) =2[A^{2}-(1+\Sigma_{+})^{2}]+\frac{3}{2}(2-\gamma
)\Omega+9\Sigma_{13}^{2},\label{2.1a}\\
\boldsymbol \partial_{1}(\tilde{N}_{22}) =2(A+3\tilde{N}_{23})\tilde{N}_{22}
+(6\tilde{\Sigma}_{22}+2-q)\tilde{\Sigma}_{23},\label{2.1b}\\
\boldsymbol \partial_{1}(\tilde{N}_{23}) =2A\tilde{N}_{23}+6\tilde{\Sigma}_{23}
^{2}-6\tilde{N}_{22}^{2}+(q-2)\tilde{\Sigma}_{22}-3\Sigma_{13}^{2},\label{2.1c}\\
\boldsymbol \partial_{1}(\tilde{\Sigma}_{23}) =-6(\tilde{\Sigma}_{23}\tilde{N}
_{23}+\tilde{N}_{22}\tilde{\Sigma}_{22}),\label{2.1d}\\
\boldsymbol \partial_{1}(\Sigma_{13})   =3(A+\dot{U}-N_{23})\Sigma_{13}.\label{2.1e} }
\end{eqnarray}
\textit{Defining Equations for $\dot{U}$ and $\Omega$}:
\numparts
\begin{eqnarray}
\fl 3(2-\gamma)\dot{U} =(3\gamma-2)A+12\tilde{\Sigma}_{23}\tilde{N}
_{22}-12\tilde{\Sigma}_{22}\tilde{N}_{23},\label{2.2a}\\
\fl 3(\gamma-1)\Omega =A(A-6\dot{U})-9\tilde{\Sigma}_{22}^{2}-9\tilde{\Sigma
}_{23}^{2}-9\tilde{N}_{22}^{2}-9\tilde{N}_{23}^{2}\nonumber\\
\fl \phantom{3(\gamma-1)\Omega=} -(1+\Sigma_{+})(1+7\Sigma_{+})+9\Sigma_{13}^{2},\label{2.2b} 
\end{eqnarray}
\endnumparts
where
\begin{eqnarray}
\fl \Sigma_{+}=\frac{-(3\gamma-2)}{4},\quad\tilde{\Sigma}_{22}=\frac
{(6-5\gamma)}{4},\quad q=\frac{(3\gamma-2)}{2}.\label{2.3}
\end{eqnarray}
\textit{Auxiliary equation:}
\begin{eqnarray}
\fl \boldsymbol \partial_{1}(\Omega)=-3\frac{(2-\gamma)}{\left(  \gamma-1\right)  }\dot
{U}\Omega.\label{2.4}
\end{eqnarray}
The vector $\mathbf{Z}=(A,\tilde{N}_{23},\tilde{N}_{22},\tilde
{\Sigma}_{23},\Sigma_{13})$ determines the dynamical state of the cosmological model. The EFEs for the models under consideration have been written as a \textbf{five-dimensional} differential equation, for the dependent variable $\mathbf{Z}$, in terms of the independent variable $\tilde{x}$, where $\boldsymbol \partial_{1}=\frac{d}{d\tilde{x}}$, with \textbf{one parameter}, $\gamma$. The system admits the two discrete symmetries:
\begin{eqnarray}
\fl (\tilde{x},A,\tilde{N}_{23},\tilde{N}_{22},\tilde{\Sigma}_{23},\Sigma_{13 }) \mapsto (\tilde{x},A,\tilde{N}_{23},\tilde{N}_{22},\tilde{\Sigma}_{23},-\Sigma_{13}) \label{2.5a}\\
\fl (\tilde{x},A,\tilde{N}_{23},\tilde{N}_{22},\tilde{\Sigma}_{23},\Sigma_{13 }) \mapsto (-\tilde{x},-A,-\tilde{N}_{23},-\tilde{N}_{22},-\tilde{\Sigma}_{23},\Sigma_{13}) \label{2.5b}
\end{eqnarray}
As a consequence of the first symmetry and the fact that $\Sigma_{13}=0$ is an invariant set, we need to only consider models with $\Sigma
_{13}\geq0$ from this point onward.

Since the expression for the energy density $\Omega$ is cubic, the vacuum part
of the boundary of the phase space is not compact and consequently the
physically relevant region of phase space is also non-compact. We have been
unable to find a particularly elegant and natural way of compactifying the
complete system, however, in section 4, we obtain a number of important
deductions about the corresponding cosmological models.\\
We refer to the physical region of phase space to mean the subset of $\Reals^{5}$ with $\Sigma_{13}\geq0$ and $\Omega\geq0$.\\
In order to evaluate the physical variables, we must first obtain the
expansion scalar $\theta$ by examination of the decoupled equations derived from (\ref{A.27}), (\ref{A.29}), and (\ref{A.20}):
\begin{eqnarray}
\boldsymbol \partial_{0}(\theta)=-(\frac{3\gamma}{2})\theta,\qquad \boldsymbol \partial_{1}(\theta)=-3\dot{U}\theta. \label{2.6}
\end{eqnarray}
Combining these with the auxiliary equation (\ref{2.4}), it is clear that we have,
\begin{eqnarray}
\theta=\theta_{0}\Omega(\tilde{x})^{\frac{\gamma-1}{2-\gamma}}e^{\frac
	{-3\gamma\tau}{2}},\mbox{ and }\mu=\mu_{0}\Omega(\tilde{x})^{\frac{\gamma
	}{2-\gamma}}e^{-3\gamma\tau},\label{2.7}
\end{eqnarray}
where $\theta_{0}$ and $\mu_{0}$ are constants. The $H_{3}$ orbits are labelled by the co-ordinate $\tilde{x}$, and the $G_{2}$ orbits are given by $\tilde{x}=$ constant, $\tau=$ constant. It follows that the energy density $\mu$ is bounded on any spatial slice $\tau=$ constant if and only if the dimensionless density parameter, $\Omega$, is bounded on that slice.\\
There are a number of invariant sets contained within this system of
equations, these are given in Table \ref{TableOfInvariantSets}.
\begin{table}[h!]
    \caption{Invariant Sets}
	\label{TableOfInvariantSets}
\begin{indented}
\item[]\begin{tabular}{@{}*{5}{l}}
		\br
		Restriction & Restriction on the Variable & Reference\\\hline
		$G_{2}$ acts OT & $\Sigma_{13}=0$ & \cite{hewitt1997dynamical}\\\hline
		$G_{2}$ has one HO KVF & $\tilde{\Sigma}_{23}=0\quad\tilde{N}_{22}=0$ & \cite{work-in-progress} \\\hline
		$G_{2}$ has two HO KVFs & $\Sigma_{13}=0\quad\tilde{\Sigma}_{23}=0\quad
		\tilde{N}_{22}=0$ & \cite{hewitt1988qualitativeI} and \cite{hewitt1991qualitativeII} \\\hline
		Vacuum & $\Omega=0$ & \\\hline
	\end{tabular}
\end{indented}
\end{table}
\section{Stability of the Equilibrium Points and Cosmological Consequences}

In Table \ref{equilibriumPointsOfSystem} we list the equilibrium points of system (\ref{2.1a}), and then we proceed to discuss their local stability and role in phase space, by examining the eigenvalues of the linearization at each equilibrium point.\\

We introduce the non-negative quantity, $s$, given by $s^{2}=(3\gamma
-2)(2-\gamma)$, and we define the quadratics $Q_{1}(\gamma),Q_{2}(\gamma)$ and
$Q_{3}({\small \Sigma}_{13})$ by 
\begin{eqnarray}
\fl Q_{1}(\gamma)=36\gamma-20-15\gamma^{2}, \  \mbox{ which is non-negative for }\gamma\in(1,\frac{3}{2}],\label{4.1a}\\
\fl Q_{2}(\gamma)=(10-7\gamma)(\gamma-1),\mbox{ and }Q_{3}({ \Sigma}
_{13})={4+81\Sigma_{13}^{2}}.\label{4.1b}
\end{eqnarray}
\begin{table}[h!]
	\caption{Equilibrium Points of the System \ref{2.1a}.}
	\label{equilibriumPointsOfSystem}
\begin{tabular}{@{}*{7}{l}}
		\br
		Name & $A$ & $\tilde{N}_{22}$ & $\tilde{N}_{23}$ & $\tilde{\Sigma}_{23}$ &
		$\Sigma_{13}^{2}$ & $\Omega$\\\hline
		Collins & $\frac{\mp3(6-5\gamma)s}{4(3\gamma-2)}$ & $0$ & $\frac{\mp s}{4} $ &
		$0$ & $0$ & $\frac{3Q_{2}(\gamma)}{3\gamma-2}$\\\hline
		Wainwright & $\frac{\pm\sqrt{Q_{3}({\small \Sigma}_{13})}}{3\sqrt{6}}$ & $0$
		& $\frac{\pm\sqrt{Q_{3}({\Sigma}_{13})}}{3\sqrt{6}}$ & $0$ & $0\leq\Sigma_{13}^{2}\leq\frac{5}{81}$ & $\frac{5-{\Sigma}_{13}^{2}
		}{90}$\\\hline
		Robinson-Trautman & $\frac{\pm\sqrt{6}{\small (3\gamma-4)}}{4}$ & $0$ &
		$\frac{\pm\sqrt{6}}{6}$ & $0$ & $\frac{Q_{1}(\gamma)}{24}$ & $0$\\\hline
		Plane-Wave & $\frac{\pm3{(2-\gamma)}}{4}$ & $\mp\tilde{\Sigma}_{23} $ &
		$\frac{\pm(6-5\gamma)}{4}$ & $\tilde{\Sigma}_{23}$ & $0$ & $0$\\\hline
	\end{tabular}
\end{table}
\subsection{The Collins Equilibrium Points.}
These equilibrium points exist in the physical region of phase space for values of the equation of state parameter satisfying $1<\gamma\leq\frac{10}{7}$. For values of $\gamma$ in this interval, but $\gamma\neq\frac{6}{5} $, it is found that the real part of three of the eigenvalues take one sign while the real part of the fourth eigenvalue takes the opposite sign, and the fifth eigenvalue is $\frac{3s(10-9\gamma)}{2(2-3\gamma)}$. We conclude that the Collins equilibrium points are saddle points when they lie in the physical region of phase space.\\

There are three values of $\gamma$ which are bifurcation values for the Collins equilibrium points. First of all, when $\gamma=\frac{10}{9}$ the fifth eigenvalue changes sign as a line bifurcation transfers (one dimension of) stability to the Robinson-Trautman equilibrium point via the Wainwright line of equilibrium points. Secondly, when $\gamma=\frac{6}{5}$, two of the eigenvalues are purely imaginary and there are closed curves on the corresponding invariant set as there is a degenerate Hopf bifurcation (see \cite{hewitt1988qualitativeI}). Two of the other eigenvalues achieve one sign while the final one achieves the opposite sign, and so these points are still saddle points in the full five-dimensional phase space. Thirdly, when $\gamma=\frac{10}{7}$, the Collins points pass through the vacuum boundary via the plane wave curve and they have two zero eigenvalues, one zero corresponds to the coallescing, and the other for coallescing with a line of equilibrium points, this is a
transcritical bifurcation.
\subsection{The Wainwright Curves of Equilibrium Points. ($\gamma=\frac{10}{9}$ only).}

When the equation of state parameter is equal to $\frac{10}{9}$, then there is
a curve of equilibrium points in phase space, given by:
\begin{eqnarray}
\fl \Sigma_{13}\in \Reals, 54A^{2}-{\small 4-81}\Sigma_{13}^{2}=0, \tilde{N}_{23}=A, \tilde{N}_{22}=\tilde{\Sigma}_{23}=0.\label{4.2}
\end{eqnarray}
The energy density along this curve is given by
\begin{eqnarray}
\fl \Omega=\frac{5-81{\Sigma}_{13}^{2}}{9}, \label{4.3}
\end{eqnarray}
and is non-negative iff $0\leq\Sigma_{13}^{2}\leq\frac{5}{81}$. We
have two finite hyperbolic arcs of Wainwright equilibrium points in the
physical region of phase space.\\
One of the five eigenvalues is zero due to the fact that there is a curve of equilibrium points. The other four are given by $[A\pm\frac{1}{3}
\sqrt{441A^{2}-12}],[A\pm\sqrt{1521A^{2}-192}]$. When $A$ is
positive then the real parts of these eigenvalues are $(+,-,+,+)$ for $0\leq A^{2}<\frac{12}{95}$, $(+,-,+,0)$ for $\frac{12}{95}=A^{2}$, and $(+,-,+,-)$ for $\frac{12}{95}<A^{2}\leq\frac{1}{6}$. When $A$ is negative then all these signs switch to the opposite value. Thus the Wainwright curves of equilibrium points are saddles in phase space. The change in sign of one of the eigenvalues arises as there is a saddle-node bifurcation occuring on the one parameter family of invariant two-spaces, see \cite{work-in-progress}.

\subsection{The Robinson-Trautman Equilibrium Points.}

At each of these equilibrium points, two of the eigenvalues are positive and two of the eigenvalues are negative. The fifth eigenvalue is proportional to the factor $(10-9\gamma)$. This is expected as each one of them lies on one of the Wainwright curves of equilibrium points when $\gamma=\frac{10}{9}$, and so it must have a zero eigenvalue at $\gamma=\frac{10}{9}$. These equilibrium points are saddles in phase space for all values of $\gamma$ in the range $(1,\frac{3}{2}]$.

\subsection{The Plane Wave Equilibrium Points.}

The eigenvalues on the lines of plane wave equilibrium points are:\\
\begin{eqnarray}
\fl 0, \mbox{\quad} \pm\frac{3}{(\gamma-1)}[4\tilde{\Sigma}_{23}^{2}-Q_{2}(\gamma
)], \mbox{\quad}\pm\frac{3}{2}[(2-\gamma)\pm8\tilde{\Sigma}_{23}i], \mbox{\quad} \pm\frac{3}{2(2-\gamma)}[Q_{1}(\gamma)-8\tilde{\Sigma}_{23}^{2}], \nonumber
\end{eqnarray}
where the leading sign is choosen to be the same sign as the sign of the variable
$A$ at the equilibrium point. The zero eigenvalue arises because the plane
waves form a line of equilibrium points. The last eigenvalue determines the stability of the plane waves relative to the invariant set $\Sigma_{13}=0$. For values of $\gamma$ in the interval $(1\frac{10}{7})$, each plane wave line possesses two segments given by
\begin{eqnarray}
-\sqrt{\frac{Q_{1}(\gamma)}{8}}<\tilde{\Sigma}_{23}<-\sqrt{\frac{Q_{2}(\gamma)}{4}}\mbox{ and }\sqrt{\frac{Q_{2}(\gamma)}{4}}<\tilde{\Sigma}
_{23}<\sqrt{\frac{3Q_{1}(\gamma)}{8}},\label{4.5}
\end{eqnarray}
on which the (real parts of the) other four eigenvalues have the same sign. For values of $\gamma$ in the interval ($\frac{10}{7},\frac{3}{2})$ each plane wave line possesses a single segment given by
\begin{eqnarray}
-\sqrt{\frac{Q_{1}(\gamma)}{8}}<\tilde{\Sigma}_{23}<\sqrt{\frac{Q_{1}(\gamma)}{8}},\label{4.6}
\end{eqnarray}
on which the (real parts of the) other four eigenvalues have the same sign.

\subsection{Cosmological Interpretation}

Although it is not currently possible to provide a complete examination of the
full five-dimensional system, we can still provide concrete statements about
the possible behaviours of some of the trajectories, and their corresponding
cosmological models. We have shown that the Collins, Wainwright and
Robinson-Trautman equilibrium points are saddle points in phase space, and so,
non of these equilibrium points can act as a source or sink for generic
trajectories. In particular, the physical features of matter domination
together with asymptotic spatial homogeneity at large spatial distance are not
generic features of the models, as they can be in some of the subclasses of
lower dimension. However, these equilibrium points may act as intermediate
asymptotes for typical trajectories.

On the other hand, the plane waves play a much more important role in phase
space. Although most of the plane wave equilibrium points are saddles, for all
values of the equation of state parameter in the interval $(1,\frac
{3}{2})$, there is always at least one segment on the line with
positive $A$ value which is a source, and there is at least one corresponding
segment on the line with negative $A$ value which is a sink. In addition, we
are aware of trajectories in the invariant set $\Sigma_{13}=0$ which flow
from the source to the sink.

It now follows, from the approximation property of orbits \cite{sibirskiui1975introduction},
that, for all values of the equation of state parameter in the interval
$(  1,\frac{3}{2} )$ there is an open set of trajectories in the
five-dimensional phase space, with the property that each one of them flows
from one plane wave equilibrium point to another plane wave equilibrium point.
Corresponding to each of these trajectories there is a well behaved
cosmological model. The cosmological model is asymptotic to a plane wave model
at large spatial distance, that is as $x \longrightarrow +\infty$ and as $x \longrightarrow -\infty$. We
expect the asymptote to be a different plane wave model in the two different
spatial directions. The cosmological model is vacuum dominated and
acceleration dominated, that is to say, the energy density becomes
insignificant at large spatial distance, in the $\tilde{x}$ direction, with the
acceleration of the $\mathbf{e}_{0}$ congruence persisting. There is only a small
volume of space in which there is non-insignificant spatial inhomogeneity and energy density. A typical density profile appears in Figure 8 of \cite{work-in-progress}.
\section{Conclusions}
In this work we have considered a class of inhomogeneous cosmological models
which possess an Abelian $G_{2}$. This class has been ignored to a large extent due
to the complication that this Abelian $G_{2}$ does not act orthogonally
transitively. It is necessary to examine such models in order to gain insight
into the properties of the general cosmological solutions of the Einstein
field equations.

Due to the complexity of the general class of such models we have imposed two
additional restrictions. The first is that the fluid flow is orthogonal to the
orbits of the Abelian $G_{2}$, a consequence of which is that the fluid is
irrotational. The second is that there is an additional symmetry of a
homothety aligned in such a way that the fluid flow is contained in the
resulting $H_{3}$ orbits. A consequence of this is that the Einstein field
equations reduce to ordinary differential equations describing the spatial
structure, while their temporal evolution is determined by the
self-similarity. The resulting models are of interest both as examples of
inhomogeneous cosmological models which do not have an orthogonally
transitively acting $G_{2}$, and as possible asymptotic states at late times for
non-self-similar models.

We have proved the existence of a well-behaved five dimensional subset of
these inhomogeneous self-similar cosmological models. These models
possess a big bang singularity at a finite time in the past. They are
asymptotically vacuum at large spatial distance where they may be approximated
by a plane wave model, however we refer to them as being acceleration
dominated since the acceleration does not tend to zero at large spatial distance.

\smallskip

We have called the models considered here \textit{exceptional} in analogy with
the exceptional orthogonal Bianchi models and it is useful to compare their
similarities and differences.\\
{\large Similarities:}
\begin{itemize}
	\item Both classes possess an Abelian $G_{2}$ which does not act orthogonally
	transitively leading to off diagonal shear terms.

	\item The fluid flow is orthogonal to the $G_{2}$ orbits and thus it is irrotational.
	
	\item The models possess an additional symmetry, either an isometry or a
	homothety. The EFEs reduce to an ODE, describing evolution in the spatially
	homogeneous case, and describing spatial structure in the self-similar case.
	
	\item In the OT case there are two parameters, the equations of state
	parameter $\gamma$ and the group parameter $h$. In the non OT case these two
	parameters are restricted: in the spatially homogeneous case we have
	$h=\frac{-1}{9}$, and in the self-similar case the restriction is the
	relation, $h=-[  2\frac{(3\gamma-2)}{(5\gamma-6)}]^{2}$.
	
	\item The class of models is five-dimensional and thus has the same size as
	the non-exceptional class, which is four-dimensional, with an
	additional free parameter.
	
	\item The following transitively self-similar solutions appear as equilibrium
	points: the Collins Bianchi $VI_{h}$ solution; the Wainwright Bianchi
	$VI_{\frac{-1}{9}}$($\gamma=\frac{10}{9})$ solution; the Robinson-Trautman
	solution; some spatially homogeneous plane wave solutions.
	
	\item Each phase space contains a one-parameter family of equilibrium points
	which play an important role for the trajectories corresponding to the other models.
\end{itemize}
{\large Differences:}
\begin{itemize}
	\item For spatially homogeneous models, the effect of more variables is to
	reduce the region of the Kasner ring which acts as a source. For example in
	Bianchi I models the whole of the Kasner ring is a source; for the
	non-exceptional Bianchi $VI_{h}$ models, only a part of the Kasner ring is a
	source; for the exceptional Bianchi $VI_{\frac{-1}{9}}$ models the additional
	shear variable has such a dramatic destabilising effect on the Kasner ring
	that no single equilibrium point is a source. All the points on the Kasner
	ring are saddle points, there are oscillations between these Kasner saddle
	points at early times and this results in a chaotic initial singularity( see, for example, \cite{wainwright2005dynamicalSystems}).
	
	\item For the self-similar models, it is the curves of plane wave equilibrium
	points which are important as asymptotes for the other models, at large
	distance. The effect of the additional shear variable is to destabilise the
	plane wave curves, and most of the equilibrium points are saddle points.
	However, the destabilisation is not exhaustive and there is a small segment
	which is either a source or sink. Thus chaotic and oscillatory
	behaviour is not necessarily the fate of these models.
	
	\item The plane waves also appear as equilibrium points for the spatially
	homogeneous models. In the non-exceptional case, plane-wave arcs are late-time
	asymptotes for generic models for certain ranges of the equation of state
	parameter. However, in the exceptional case the additional shear variable
	completely destabilises these arcs and the late time asymptote becomes the
	Robinson-Trautman equilibrium point. Thus the Robinson-Trautman model plays a
	very important role as the only possible late-time asymptote for the generic
	exceptional models. However in the analysis of the self-similar models, the
	Robinson-Trautman equilibrium point is only a saddle point. It acts as an intermediate
	spatial asymptote for some trajectories, but most trajectories are unaffected by its existence: it only plays a minor role in phase space.
\end{itemize}

\smallskip

As far as the Chaotic Inhomogeneity Conjecture is concerned, we have taken a
first step in proving this result by providing a class of inhomogeneous
cosmological models in which the plane wave curves are partially destabilised.
It is likely that by examining a slightly broader class of related models,
e.g. by introducing one-dimension of tilt, it will be possible
to complete this destabilisation.

\appendix

\section{The EFEs For The Irrotational N-OT $G_{2}$ Cosmologies}

In this appendix we derive the EFEs for perfect fluid spacetimes with a linear equation of state, and which admit an Abelian two dimensional isometry group with space-like orbits and where fluid flow is orthogonal to these orbits.  We then dimensionalize the equations. The
dynamical equilibrium states are defined to be the equilibrium points of the
PDE. These equilibrium points correspond to spacetimes which admit, in
addition, a homothetic vector field, aligned so that fluid flow is tangential
to the resulting $H_{3}$ orbits. The ODE for the dynamical equilibrium states are also derived here.\\
When spacetime admits a $G_{2}$, with infinitesimal generators $\boldsymbol \xi$ and $\boldsymbol \eta$, it is
possible \cite{WainWright1979Appen} to select a group-invariant orthonormal frame:
\begin{eqnarray}
\lbrack \mathbf{e}_{a}, \boldsymbol \xi]=0,\qquad\lbrack \mathbf{e}_{a}, \boldsymbol \eta]=0.\label{A.1}
\end{eqnarray}
It is also possible to further align the frame so that $\mathbf{e}_{2}$ and $\mathbf{e}_{3}$ are
tangential to the group orbits (orbit-aligned), and then the commutation
functions ($\gamma_{\;bc}^{a}$) of this group-invariant orbit-aligned frame
are restricted by
\begin{eqnarray}
\mathbf{e}_{A}(\gamma_{\;bc}^{a})=0.\label{A.2a}
\end{eqnarray}
Furthermore, we can, in addition, align the frame so that $\mathbf{e}_{0}=\mathbf{u}$, and refer to
the frame as fluid aligned. Any unit vector which is both invariant under the
action of the $G_{2}$, and is orthogonal to the group orbits,
is HO \cite[p.1134]{wainwright1981exact}. That is both $\mathbf{e}_{0}$ and $\mathbf{e}_{1}$ are HO, and
the commutation functions are also restricted by:
\begin{eqnarray}
\gamma_{\ 0A}^{0}=0,\gamma_{\ 1A}^{0}=0,\gamma_{\ 0A}^{1}=0,\gamma_{\ 1A}
^{1}=0.\label{A.2b}
\end{eqnarray}
Since we have an Abelian $G_{2}$, that is $\boldsymbol \xi$ and $ \boldsymbol \eta$ commute,
it follows that the two vectors that are aligned with the $G_{2}$ orbits,
$\mathbf{e}_{2}$ and $\mathbf{e}_{3}$, commute:%
\begin{eqnarray}
\gamma_{\;23}^{a}=0. \label{A.2c}
\end{eqnarray}
The remaining non-discrete frame-freedom is a rotation in the $G_{2}$ orbits.
The non-zero commutation functions may be expressed in terms of the kinematic
and spatial curvature variables of \cite{Cargesemaccallum1973cosmological} and are:
\begin{eqnarray}
\theta,\quad\sigma_{ab},\quad\dot{u}_{1},\quad a_{1},\quad n_{AB},\quad
\Omega_{a},\label{A.3}
\end{eqnarray}
where,
\begin{eqnarray}
\Omega_{2}=\sigma_{13},\quad\Omega_{3}=-\sigma_{12}.\label{A.4}
\end{eqnarray}
At this point we follow \cite{hewitt1990orthogonally} and introduce a two-dimensional tensor 
formalism to distinguish between quantities which transform as scalars,
vectors or rank 2 tensors under the rotational freedom remaining. We define:
\begin{eqnarray}
n_{+} =\frac{3}{2}n_{\;C}^{C},\;\tilde{n}_{AB}=n_{AB}-\frac{1}{3}
n_{+}\delta_{AB},\;^{\ast}\tilde{n}_{AB}=\tilde{n}_{\;A}^{C}\epsilon
_{BC},\label{A.5a}\\
\sigma_{+} =\frac{3}{2}\sigma_{\;C}^{C},\;\tilde{\sigma}_{AB}=\sigma
_{AB}-\frac{1}{3}\sigma_{+}\delta_{AB},\;^{\ast}\tilde{\sigma}_{AB}
=\tilde{\sigma}_{\;A}^{C}\epsilon_{BC},\label{A.5b}\\
^{\ast}\sigma_{1A} =\varepsilon_{AB}\sigma_{\;1}^{B}\;(=\Omega
_{A}),\label{A.5c}
\end{eqnarray}
where $\varepsilon_{AB}$ is the two-dimensional permutation symbol.
Under the rotational frame-freedom, the variables $\theta,\sigma_{+},\dot
{u}_{1}$, and $a_{1}$ transform as scalars, the variables $\ \sigma_{1A}$
($\Omega_{A}$) transform as the components of a vector and the trace-free
symmetric quantities $\tilde{\sigma}_{AB}$ and $\tilde{n}_{AB}$ tranform as
rank two tensors. The quantities $n_{+}$ and $\Omega_{1}$ are non-tensorial
and transform according to
\begin{eqnarray}
n_{+}=n_{+}+3\mathbf{e}_{1}(\phi),\quad\quad\Omega_{1}=\Omega_{1}+\mathbf{e}_{0}(\phi),\label{A.6}
\end{eqnarray}
where $\phi$ is the angle of rotation of the frame. A 1+1+2 decompositon of the EFEs can now be performed in an analogous manner to that done in \cite{hewitt1990orthogonally}. The evolution equation for $\dot{u}_{1}$, valid for $\mu\neq0$, is obtained by applying the commutator $[\mathbf{e}_{0},\mathbf{e}_{1}]$ to $\mu$.\\ \\
\textit{System of PDEs in the physical variables}\\
Evolution Equations:
\begin{eqnarray}
\fl \mathbf{e}_{0}(\theta)  =-\frac{1}{3}\theta^{2}-\frac{2}{3}\sigma_{+}^{2}
-\tilde{\sigma}^{AB}\tilde{\sigma}_{AB}+(\dot{u}_{1}-2a_{1})\dot{u}_{1}
+ \mathbf{e}_{1} (\dot{u}_{1}) \nonumber\\
\fl \phantom{\mathbf{e}_{0}(\theta)=} -\frac{1}{2}({3\gamma-2})\mu-2\sigma_{1C}\sigma^{1C}, \label{A.7a}\\
\fl \mathbf{e}_{0}(\sigma_{+})  =-\theta\sigma_{+}-\tilde{n}^{AB}\tilde{n}_{AB}-(\dot
{u}_{1}+a_{1})\dot{u}_{1}-\mathbf{e}_{1}(\dot{u}_{1})+\mathbf{e}_{1}(a_{1}) \nonumber\\
\fl \phantom{\mathbf{e}_{0}(\sigma_{+})=} +3\sigma_{1C}\sigma^{1C}, \label{A.7b}
\end{eqnarray}
\begin{eqnarray}
\fl \mathbf{e}_{0}(\tilde{\sigma}_{AB})  =-\theta\tilde{\sigma}_{AB}+2\Omega_{1}^{\ast
}\tilde{n}_{AB}-\frac{2}{3}n_{+}\tilde{n}_{AB}+(2a_{1}-\dot{u}_{1})\,^{\ast}\tilde{n}_{AB} \nonumber\\
\fl \phantom{\mathbf{e}_{0}(\tilde{\sigma}_{AB})=} -\mathbf{e}_{1}(^{\ast}\tilde{n}_{AB})+2\sigma_{1A}\sigma_{1B}-\sigma_{1C}\sigma^{1C}\delta_{AB}, \label{A.7c}\\
\fl \mathbf{e}_{0}(a_{1})  =\frac{1}{3}(2\sigma_{+}-\theta)a_{1}-\frac{1}{3}(\sigma
_{+}+\theta)\dot{u}_{1}-\frac{1}{3}\mathbf{e}_{1}(\sigma_{+}+\theta), \label{A.7d}\\
\fl \mathbf{e}_{0}(n_{+})  =\frac{1}{3}(2\sigma_{+}-\theta)n_{+}+3\tilde{\sigma}
^{AB}\tilde{n}_{AB}+3\dot{u}_{1}\Omega_{1}+3\mathbf{e}_{1}(\Omega_{1}), \label{A.7e}\\
\fl \mathbf{e}_{0}(\tilde{n}_{AB})  =\frac{1}{3}(2\sigma_{+}-\theta)\tilde{n}
_{AB}+2\Omega_{1}\,^{\ast}\tilde{n}_{AB}+\frac{2}{3}n_{+}\tilde{\sigma}
_{AB} \nonumber\\ 
\fl \phantom{\mathbf{e}_{0}(\tilde{n}_{AB})=} +\dot{u}_{1}\,^{\ast}\tilde{\sigma}_{AB}+\mathbf{e}_{1}(^{\ast}\tilde{\sigma}_{AB}), \label{A.7f}\\
\fl \mathbf{e}_{0}(\dot{u}_{1})  =\frac{1}{3}(2\sigma_{+}-\theta)\dot{u}_{1}+(\gamma-1)(\theta\dot{u}_{1}+\mathbf{e}_{1}(\theta)),\label{A.7g}\\
\fl \mathbf{e}_{0}(\sigma_{1A}) =-(\theta+\sigma_{+})\sigma_{1A}+\Omega_{1}\,^{\ast}\sigma_{1A}-\tilde{\sigma}_{\quad A}^{C}\sigma_{1C}.\label{A.7h}
\end{eqnarray}
Constraint Equations:
\begin{eqnarray}
\mathbf{e}_{1}(\sigma_{+}+\theta) =3a_{1}\sigma_{+}-\frac{3}{2}^{\ast}\tilde
{\sigma}^{AB}\tilde{n}_{AB},\label{A.8a}\\
\mathbf{e}_{1}(\sigma_{1A}) =3a_{1}\sigma_{1A}+\,^{\ast}\tilde{n}_{\quad A}
^{B}\sigma_{1B}+\frac{1}{3}n_{+}{}^{\ast}\sigma_{1A}.\label{A.8b}
\end{eqnarray}
Defining Equation for $\mu$:
\begin{eqnarray}
\fl \mathbf{e}_{1}(a_{1}) =\frac{3}{2}a_{1}^{2}-\frac{1}{6}\theta^{2}+\frac{1}
{6}\sigma_{+}^{2}+\frac{1}{4}\tilde{\sigma}^{AB}\tilde{\sigma}_{AB}
+\frac{1}{4}\tilde{n}^{AB}\tilde{n}_{AB}+\frac{1}{2}\mu+\frac{1}{2}
\sigma^{1A}\sigma_{1A}.\label{A.9}
\end{eqnarray}
Auxiliary Equations:
\begin{eqnarray}
\mathbf{e}_{0}(\mu) =-\gamma\mu\theta, \label{A.10a}\\
(\gamma-1)\mathbf{e}_{1}(\mu) =-\gamma\mu\dot{u}_{1}.\label{A.10b}
\end{eqnarray}
Commutator of $\mathbf{e}_{0}$ and $\mathbf{e}_{1}:$%
\begin{eqnarray}
\lbrack \mathbf{e}_{0},\mathbf{e}_{1}]=\dot{u}_{1}\mathbf{e}_{0}+\frac{1}{3}(2\sigma_{+}-\theta
)\mathbf{e}_{1}-2\sigma_{12}\mathbf{e}_{2}-2\sigma_{13}\mathbf{e}_{3}.\label{A.11}
\end{eqnarray}
Dimensionless variables are now defined by normalising with the rate of
expansion scalar, $\theta$:
\begin{eqnarray}
\Sigma_{+} =\frac{\sigma_{+}}{\theta},\;\tilde{\Sigma}_{AB}=\frac
{\tilde{\sigma}_{AB}}{\theta},\;\Sigma_{1A}=\frac{\sigma_{1A}}{\theta}
,\;\dot{U}=\frac{\dot{u}_{1}}{\theta},\;R=\frac{3\Omega_{1}}{\theta
}, \label{A.12a}\\
A =\frac{3a_{1}}{\theta},\;N_{+}=\frac{n_{+}}{\theta},\;\tilde{N}
_{AB}=\frac{\tilde{n}_{AB}}{\theta},\;\Omega=\frac{3\mu}{\theta^{2}
}.\label{A.12b}
\end{eqnarray}
We must also introduce two dimensionless scalars, $q$ and $r$, which are
formed from the derivatives of the expansion scalar:
\begin{eqnarray}
\mathbf{e}_{0}(\theta)=-\frac{1}{3}(1+q)\theta^{2},\quad\quad \mathbf{e}_{1}(\theta)=-\frac
{1}{3}r\theta^{2},\label{A.13}
\end{eqnarray}
and the dimensionless differential operators:
\begin{eqnarray}
\boldsymbol \partial_{a}=\frac{3}{\theta}\mathbf{e}_{a}.\label{A.14}
\end{eqnarray}
\textit{System of PDEs in dimensionless Variables}
Evolution Equations
\begin{eqnarray}
\fl \boldsymbol \partial_{0}(\Sigma_{+})   =(q-2)\Sigma_{+}-3\tilde{N}_{AB}\tilde{N}
^{AB}-(3\dot{U}-r+A)\dot{U} \nonumber \\
\fl \phantom{\partial_{0}(\Sigma_{+})=} -\frac{1}{3}rA-\boldsymbol \partial_{1}(\dot{U})+\frac{1}{3}\boldsymbol \partial_{1}(A)+9\Sigma_{1A}\Sigma^{1A}, \label{A.15a}\\
\fl \boldsymbol\partial_{0}(\tilde{\Sigma}_{AB})=(q-2)\tilde{\Sigma}_{AB}+2R\,^{\ast}\tilde{\Sigma}_{AB}-2N_{+}\tilde{N}_{AB}-\boldsymbol \partial_{1}(\,^{\ast}\tilde{N}_{AB}) \nonumber \\
\fl \phantom{\partial_{0}(\tilde{\Sigma}_{AB})=} -(3\dot{U}-r-2A)\,^{\ast}\tilde{N}_{AB}+6\Sigma_{1A}\Sigma_{1B}-3\Sigma_{1C}\Sigma^{1C}\delta_{AB}, \label{A.15b}\\
\fl \boldsymbol\partial_{0}(A)=(q+2\Sigma_{+})A-(3\dot{U}-r)(1+\Sigma_{+})-\partial_{1}(\Sigma_{+}),\label{A.15c}\\
\fl \boldsymbol\partial_{0}(N_{+})=(q+2\Sigma_{+})N_{+}+9\tilde{\Sigma}^{AB}\tilde{N}_{AB}+(3\dot{U}-r)R+\boldsymbol \partial_{1}(R),\label{A.15d}\\
\fl \boldsymbol\partial_{0}(\tilde{N}_{AB})=(q+2\Sigma_{+})\tilde{N}_{AB}+2R\,^{\ast}\tilde{N}_{AB}+2N_{+}\tilde{\Sigma}_{AB} \nonumber\\ 
\fl \phantom{\partial_{0}(\tilde{N}_{AB})=} +(3U-r)\,^{\ast}\tilde{\Sigma}_{AB}+\boldsymbol\partial_{1}(\,^{\ast}\tilde{\Sigma}_{AB}),\label{A.15e}\\
\fl \boldsymbol \partial_{0}(\dot{U})=(q+2\Sigma_{+})\dot{U}+(\gamma-1)(3\dot
{U}-r),\label{A.15f}\\
\fl \boldsymbol\partial_{0}(\Sigma_{1A})=(q-2-3\Sigma_{+})\Sigma_{1A}+R\,^{\ast}\Sigma_{1A}-3\tilde{\Sigma}_{\quad A}^{C}\Sigma_{1C}.\label{A.15g}
\end{eqnarray}
Defining equations for $r$, $\Omega$ and $q$:
\begin{eqnarray}
\fl \boldsymbol \partial_{1}(\Sigma_{+})  =3A\Sigma_{+}-\frac{9}{2}\,^{\ast}\tilde{\Sigma
}^{AB}\tilde{N}_{AB}+r(1+\Sigma_{+}),\label{A.16a}\\
\fl \boldsymbol \partial_{1}(A) =(\frac{3}{2}A+r)A+\frac{9}{4}(\tilde{\Sigma}^{AB}
\tilde{\Sigma}_{AB}+\tilde{N}^{AB}\tilde{N}_{AB})+\frac{3}{2}\Sigma_{+}^{2}-\frac{3}{2}+\frac{3}{2}\Omega+\frac{9}
{2}\Sigma_{1A}\Sigma^{1A},\label{A.16b}\\
\fl \boldsymbol \partial_{1}(\dot{U}) =-(3\dot{U}-r-2A)\dot{U}+2\Sigma_{+}^{2}
+3\tilde{\Sigma}^{AB}\tilde{\Sigma}_{AB}+\frac{1}{2}(3\gamma-2)\Omega-q+6\Sigma_{1A}\Sigma^{1A}. \label{A.16c}
\end{eqnarray}
Constraint Equation: 
\begin{eqnarray}
\boldsymbol \partial_{1}(\Sigma_{1A})=(3A+r)\Sigma_{1A}+3\,^{\ast}\tilde{N}_{A}^{\quad B}\Sigma_{1B}+N_{+}\,^{\ast}\Sigma_{1A}.\label{A.17}
\end{eqnarray}
Auxiliary Equations:
\begin{eqnarray}
\boldsymbol \partial_{0}(r)-\boldsymbol \partial_{1}(q)=(3\dot{U}-r)(1+q)+(q+2\Sigma_{+})r, \label{A.18a}\\
\boldsymbol \partial_{0}(\Omega)=[2q-(3\gamma-2)]\Omega, \label{A.18b}\\
(\gamma-1)\boldsymbol \partial_{1}(\Omega) =[2(\gamma-1)r-3\gamma\dot{U}]\Omega. \label{A.18c}
\end{eqnarray}
Commutator of $\boldsymbol \partial_{0}$and $\boldsymbol \partial_{1}:$
\begin{eqnarray}
\lbrack\boldsymbol \partial_{0},\boldsymbol \partial_{1}]=(3\dot{U}-r)\boldsymbol \partial_{0}+(q+2\Sigma
_{+})\boldsymbol \partial_{1}-6\Sigma_{12}\boldsymbol \partial_{2}-6\Sigma_{13}\boldsymbol \partial_{3}.\label{A.19}
\end{eqnarray}
Decoupled Equations:
\begin{eqnarray}
\boldsymbol \partial_{0}(\theta)=-(1+q)\theta,\qquad\boldsymbol \partial_{1}(\theta)=-r\theta.\label{A.20}
\end{eqnarray}
We consider the vector
\begin{eqnarray}
\mathbf{X}=(\Sigma_{+},\tilde{\Sigma}_{AB},A,N_{+},\tilde{N}_{AB},\dot
{U},\Sigma_{1A})\label{A.21}
\end{eqnarray}
as providing the dynamical state of a model.\\
It is pertinant at this point to note that were we to restrict the equations by $\boldsymbol \partial_{1}\mathbf{X}=\mathbf{0,}$ then it follows from (\ref{A.19}) and (\ref{A.18a})-(\ref{A.18c}) that the corresponding models must satisfy one of the following:
\begin{enumerate}
	\item $\boldsymbol \partial_{0}\mathbf{X}=\mathbf{0,}$ in this case the models are irrotational transitively self-similar. These are all well-known exact solutions (see \cite{hsu1986self}), and they arise as equilibrium points in the models which we examine.
	
	\item $\gamma=2$, in this case the models are spatially self-similar and stiff
	- these have been discussed at length by various authors (see \cite{mcintosh1976homothetic}
	and \cite{kramer1980exact} section 23.1).
	
	\item  $\Omega=0$, in this case the models are both spatially self-similar and vacuum.
	
	\item $\dot{U}=0$, in this case the models are orthogonally spatially homogeneous of Bianchi type I-VII. These models are separated into non-exceptional (orthogonally-transitive), and exceptional (non-orthogonally-transitive), according to the vanishing, or non-vanishing, of $\Sigma_{1A}$. In the latter case it is the \textit{constraint equation for} $\Sigma_{1A}$ (\ref{A.17}) which forces the well known restriction on the group
	parameter, namely $h=\frac{-1}{9}$. In the models considered below it is the \textit{evolution equation for }$\Sigma_{1A}$ (\ref{A.15g}) which yields the analogous restriction discussed after equation (\ref{A.35}).
\end{enumerate}
The dynamical equilibrium states of the $G_{2}$ cosmologies are defined (see
\cite{hewitt1990orthogonally}) by imposing the additional condition of
\begin{eqnarray}
\boldsymbol \partial_{0}\mathbf{X}=\mathbf{0}.\label{A.22}
\end{eqnarray}
It follows from the EFEs that
\begin{eqnarray}
\boldsymbol \partial_{0}\boldsymbol \partial_{1}\mathbf{X}=\mathbf{0,}\label{A.23}
\end{eqnarray}
almost everywhere.The defining equations for $r,q$ and $\Omega$ now yield
\begin{eqnarray}
\boldsymbol \partial_{0}r=0,\qquad\boldsymbol \partial_{0}\Omega=0,\qquad\boldsymbol \partial_{0}q=0,\label{A.24}
\end{eqnarray}
and the commutator of $\boldsymbol \partial_{0}$ and $\boldsymbol \partial_{1}$ now implies that
\begin{eqnarray}
(q+2\Sigma_{+})\boldsymbol \partial_{1} \mathbf{X}=0,\label{A.25}
\end{eqnarray}
and so there are the two possibilities: either
\begin{eqnarray}
& \boldsymbol \partial_{1} \mathbf{X}=0,\label{A.26a}\\
\mbox{or } & \boldsymbol \partial_{1} \mathbf{X} \neq 0,\qquad(q+2\Sigma_{+})=0.\label{A.26b}
\end{eqnarray}
Models which satisfy the former condition have been discussed above, see the
section after equation (\ref{A.21}). We now consider those models which
satisfy the latter conditions, which is our primary intent, these are referred to as the spatially
inhomogeneous dynamical equilibrium states of the $G_{2}$ cosmologies. For
non-vacuum models we can use the $\boldsymbol \partial_{0}\dot{U}$ equation to deduce that
\begin{eqnarray}
3\dot{U}-r=0,\label{A.27}
\end{eqnarray}
and the $\boldsymbol \partial_{0}A$ equation now yields:
\begin{eqnarray}
\boldsymbol \partial_{1}\Sigma_{+}=0.\label{A.28}
\end{eqnarray}
Furthermore, the $\boldsymbol \partial_{0}\Omega$ equation now implies that:
\begin{eqnarray}
q=\frac{1}{2}(3\gamma-2),\quad\Sigma_{+}=-\frac{1}{4}(3\gamma-2).\label{A.29}
\end{eqnarray}
The following proposition characterises the spatially inhomogeneous dynamical equilibrium states as being self-similar:
\begin{prop}
	An irrotational N-OT $G_{2}$ cosmology is a spatially inhomogeneous dynamical equilibrium state if and only if the spacetime is self-similar, admitting a maximal $H_{3}$ acting on the hypersurfaces generated by the two Killing vector fields and the fluid 4-velocity. 
\end{prop}
\begin{proof}
	This result have been proven for the Orthogonally Transitive $G_{2}$ Cosmologies in \cite{hewitt1990orthogonally}, and the same proof may be used in the non-orthogonally transitive case.	
\end{proof}
The third set of equations which we provide is for the spatially inhomogeneous
dynamical equilibrium states of the $G_{2}$ cosmologies. These equations
consist of spatial structure equations, which involve the spatial derivative
operator $\boldsymbol \partial_{1}$, and constraint equations.\\ \\
\textit{Spatial structure equations:}
\begin{eqnarray}
\fl \boldsymbol \partial_{1}(\dot{U})=\frac{1}{3}[A^{2}-(1+\Sigma_{+})^{2}]+\frac{3}
{2}\tilde{\Sigma}^{AB}\tilde{\Sigma}_{AB}-\frac{3}{2}\tilde{N}^{AB}\tilde{N}_{AB}+\frac{1}{2}\gamma\Omega+9\Sigma_{1A}\Sigma^{1A},\label{A.30a}\\
\fl \boldsymbol\partial_{1}(A)=2[A^{2}-(1+\Sigma_{+})^{2}]+\frac{3}{2}(2-\gamma)\Omega+9\Sigma_{1A}\Sigma^{1A},\label{A.30b}\\
\fl \boldsymbol\partial_{1}(^{\ast}\tilde{N}_{AB})=(q-2)\tilde{\Sigma}_{AB}+2R\,^{\ast}\tilde{\Sigma}_{AB}-2N_{+}\tilde{N}_{AB} \nonumber\\
\fl \phantom{\partial_{1}(^{\ast}\tilde{N}_{AB})}+2A\,^{\ast}\tilde{N}_{AB}+6\Sigma_{1A}\Sigma_{1B}-3\Sigma_{1C}\Sigma
^{1C}\delta_{AB},\label{A.30c}\\
\fl \boldsymbol\partial_{1}(R)=-9\tilde{\Sigma}^{AB}\tilde{N}_{AB}\label{A.30d},\\
\fl \boldsymbol\partial_{1}(^{\ast}\tilde{\Sigma}_{AB})=-2R\,^{\ast}\tilde{N}_{AB}-2N_{+}\tilde{\Sigma}_{AB},\label{A.30e}\\
\fl \boldsymbol\partial_{1}(\Sigma_{1A})=3(A+\dot{U})\Sigma_{1A}+3\,^{\ast}\tilde{N}_{\;A}^{B}\Sigma_{1B}+N_{+}\,^{\ast}\Sigma_{1A}.\label{A.30f}
\end{eqnarray}
Defining equations for $\Omega$:
\begin{eqnarray}
\fl \Omega =\frac{1}{3(\gamma-1)}[A(A-6\dot{U})-\frac{9}{2}\tilde{N}
^{AB}\tilde{N}_{AB}-\frac{9}{2}\tilde{\Sigma}^{AB}\tilde{\Sigma}_{AB}-(1+\Sigma_{+})(1+7\Sigma_{+}) \nonumber  \\
+9\Sigma_{1A}\Sigma^{1A}]. \label{A.31}
\end{eqnarray}
Constraint Equations:
\begin{eqnarray}
0  & =(q-2-3\Sigma_{+})\Sigma_{1A}+R\;^{\ast}\Sigma_{1A}-3\tilde{\Sigma
}_{\quad A}^{C} \tilde{\Sigma}_{1C},\label{A.32a}\\
0  & =A\Sigma_{+}-\frac{3}{2}^{\ast}\tilde{\Sigma}^{AB}\tilde{N}_{AB}+\dot
{U}(1+\Sigma_{+}).\label{A.32b}
\end{eqnarray}
Auxiliary Equations:
\begin{eqnarray}
(\gamma-1)\boldsymbol \partial_{1}(\Omega)=3\dot{U}(\gamma-2)\Omega.\label{A.33}
\end{eqnarray}
The system admits a first integral, $k$, given by:
\begin{eqnarray}
k=\frac{9}{2}\tilde{\Sigma}^{AB}\tilde{\Sigma}_{AB}-R^{2}.\label{A.34}
\end{eqnarray}
Constraint equation (\ref{A.32a}) leads to two possibilities. The first of these is that
\begin{eqnarray}
\Sigma_{12}=\Sigma_{13}=0,\label{A.35}
\end{eqnarray}
in which case the $G_{2}$ acts orthogonally transitively (these models have
been considered before \cite{hewitt1997dynamical}). The second possibility is that
\begin{eqnarray}
k =\frac{9}{16}(6-5\gamma)^{2},\label{A.36a}\\
0 =[\frac{3}{4}(5\gamma-6)-3\tilde{\Sigma}_{22}]\Sigma_{12}+(R-3\tilde
{\Sigma}_{23})\Sigma_{13}.\label{A.36b}
\end{eqnarray}
We concentrate on this second class of models in this paper. For these
models the Abelian $G_{2}$ does not necessarily act orthogonally transitively,
equivalently equation (\ref{A.35}) does not hold. Instead there is a restriction on
the parameters in the model, namely $k$ and $\gamma$ are related by expression
(\ref{A.36a}), and, in addition, the two additional shear variables are related by
expression (\ref{A.36b}). Since their EFEs reduce to a five-dimensional ODE with a single parameter ($\gamma$), they are as general as the non-exceptional parallel self-similar cosmologies, whose EFEs reduce to a four-dimensional ODE with two parameters ($k$ and $\gamma$) see \cite{hewitt1997dynamical} for more details. There is exactly the same relationship between the exceptional and the non-exceptional orthogonal Bianchi cosmologies, including the origin and imposition of the constraint, which is, of course, $h=-\frac{1}{9}$, in the exceptional case.\\
The Bianchi type of the $H_{3}$ is given by the Table \ref{TableofBianchiTypeofH3}.
\begin{table}[h!]
	\caption{The Bianchi Type of the Models.}
	\label{TableofBianchiTypeofH3}
\begin{tabular}{@{}*{9}{l}}
	\br
	Bianchi type & $I$ & $II$ & $VI_{0}$ & $VII_{0}$ & $V$ & $IV$ & $VI_{h}$ &
	$VII_{h}$\\\hline
	$\Sigma_{+}=0$ (equivalently $\gamma=\frac{2}{3}$) & $0$ & $0$ & $0$ &
	$0$ & $\neq0$ & $\neq0$ & $\neq0$ & $\neq0$\\\hline
	$\tilde{\Sigma}^{AB}\tilde{\Sigma}_{AB}$ & $0$ & $\neq0$ &  &  & $0$ &  &  &
	\\\hline
	$R$ & $0$ & $\neq0$ &  &  & $0$ &  &  & \\\hline
	$k$ & $0$ & $0$ & $<0$ & $>0$ & $0$ & $0$ & $<0$ & $>0$\\\hline
	\end{tabular}
\end{table}
The value of the group parameter $h$ is related to the first integral $k$ by
\begin{eqnarray}
h=\frac{-36\Sigma_{+}^{2}}{k}.\label{A.37}
\end{eqnarray}
Since we impose the restriction (\ref{A.36a}), and the value of $\Sigma_{+}$ is
given in (\ref{A.29}), it follows that this exceptional class of models is of
Bianchi type $VI_{h}$ with
\begin{eqnarray}
h=-\left[  2\frac{(3\gamma-2)}{(5\gamma-6)}\right]  ^{2}, \label{A.38}
\end{eqnarray}
(n.b. this is the Bianchi type of the $H_{3}$ ). Thus, the exceptional class of
self-similar models is forced to admit a restriction which connects the
equation of state parameter and the Bianchi type of the $H_{3}$. The existence
of such a restriction is an analogous situation that occurs for the exceptional
Bianchi models, however in that case the restriction is much simpler
($h=-\frac{1}{9})$.
In the paper of \cite{hewitt1997dynamical}, studying the orthogonally transitively acting
subclass, a parameter $k^{\ast}$ was defined as%
\begin{eqnarray}
k^{\ast}=-(1+\Sigma_{+})(1+7\Sigma_{+})=\frac{9}{16}(2-\gamma)(7\gamma
-6),\label{A.39}
\end{eqnarray}
and it was explained that well behaved models arise when the combination
$k-k^{\ast}$ is non-negative. We note that this combination takes the
following value here,
\begin{eqnarray}
k^{\ast}-k=9(\gamma-1)(3-2\gamma),\label{A.40}
\end{eqnarray}
and we assume that this quantity is non-negative throughout.

\ack
The author is happy to thank S. Rashidi for comments on this work and for the typesetting of this document.

\section*{References}
\bibliographystyle{iopart-num}
\bibliography{references1}

\providecommand{\newblock}{}
\begin{thebibliography}{10}
\expandafter\ifx\csname url\endcsname\relax
  \def\url#1{{\tt #1}}\fi
\expandafter\ifx\csname urlprefix\endcsname\relax\def\urlprefix{URL }\fi
\providecommand{\eprint}[2][]{\url{#2}}

\bibitem{hewitt1990orthogonally}
Hewitt C and Wainwright J 1990 {\em Classical and Quantum Gravity\/} {\bf 7}
  2295

\bibitem{hewitt1997dynamical}
Hewitt C 1997 {\em Classical and Quantum Gravity\/} {\bf 14} 3073

\bibitem{feinstein1989new}
Feinstein A and Senovilla J~M 1989 {\em Classical and Quantum Gravity\/} {\bf
  6} L89

\bibitem{senovilla1990new}
Senovilla J~M 1990 {\em Physical Review Letters\/} {\bf 64} 2219

\bibitem{wainwright1980some}
Wainwright J and Goode S 1980 {\em Physical Review D\/} {\bf 22} 1906

\bibitem{work-in-progress}
Rashidi S, Hewitt C and Charbonneau B 2021 {\em
  \href{https://arxiv.org/abs/2103.16428}{arxiv:2103.16428}\/} Dynamical
  Equilibrium States of a Class of Irrotational Non-Orthogonally Transitive
  ${G}_2$ Cosmologies II: Models With One Hypersurface Orthogonal Killing
  Vector Field

\bibitem{eardley1974commun}
Eardley D 1974 {\em Phys Rev D\/} {\bf 19} 2239

\bibitem{Cargesemaccallum1973cosmological}
MacCallum M 1973 Cosmological models from a geometric point of view {\em
  Cargese lectures in physics. Vol. 6\/}

\bibitem{wainwright2005dynamicalSystems}
Wainwright J and Ellis G 2005 {\em Dynamical systems in cosmology\/} (Cambridge
  University Press)

\bibitem{wainwright1981exact}
Wainwright J 1981 {\em Journal of Physics A: Mathematical and General\/} {\bf
  14} 1131

\bibitem{hewitt1988qualitativeI}
Hewitt C, Wainwright J and Goode S 1988 {\em Classical and Quantum Gravity\/}
  {\bf 5} 1313

\bibitem{hewitt1991qualitativeII}
Hewitt C, Wainwright J and Glaum M 1991 {\em Classical and Quantum Gravity\/}
  {\bf 8} 1505

\bibitem{krasinski1996overview}
Krasinski A 1996 {\em Recent Developments In Gravitation And Mathematical
  Physics-Proceedings Of The First Mexican School On Gravitation And
  Mathematical Physics\/}  163

\bibitem{hewitt2003asymptotic}
Hewitt C, Horwood J and Wainwright J 2003 {\em Classical and Quantum Gravity\/}
  {\bf 20} 1743

\bibitem{hewitt1991algebraic}
Hewitt C 1991 {\em General relativity and gravitation\/} {\bf 23} 1363--1383

\bibitem{uggla1989asymptotic}
Uggla C 1989 {\em Classical and Quantum Gravity\/} {\bf 6} 383

\bibitem{hewitt1997bianchi}
Hewitt C and Wainwright J 1997 Bianchi cosmologies: non-tilted class b models.
  {\em dsc\/} pp 153--169

\bibitem{collins1971more}
Collins C~B 1971 {\em Communications in Mathematical Physics\/} {\bf 23}
  137--158

\bibitem{hewitt2012conjecture}
Hewitt C 2012 The conjecture of chaotic cosmological inhomogeneity {\em The
  Twelfth Marcel Grossmann Meeting: On Recent Developments in Theoretical and
  Experimental General Relativity, Astrophysics and Relativistic Field Theories
  (In 3 Volumes)\/} (World Scientific) pp 1356--1358

\bibitem{cautun2014evolution}
Cautun M, Van De~Weygaert R, Jones B~J and Frenk C~S 2014 {\em Monthly Notices
  of the Royal Astronomical Society\/} {\bf 441} 2923--2973

\bibitem{coley2005dynamical}
Coley A and Hervik S 2005 {\em Classical and Quantum Gravity\/} {\bf 22} 579

\bibitem{sibirskiui1975introduction}
{Sibirski{\u\i}} K 1975 {\em Introduction to topological dynamics\/} (Noordhoff
  International Pub)

\bibitem{WainWright1979Appen}
Wainwright J 1979 {\em Journal of Physics A: Mathematical and General\/} {\bf
  12} 2015

\bibitem{hsu1986self}
Hsu L and Wainwright J 1986 {\em Classical and Quantum Gravity\/} {\bf 3} 1105

\bibitem{mcintosh1976homothetic}
McIntosh C 1976 {\em General Relativity and Gravitation\/} {\bf 7} 199--213

\bibitem{kramer1980exact}
Kramer D, Stephani H, MacCallum M and Herlt E 1980 {\em Exact solutions of
  Einstein’s field equations\/} (Berlin)

\end{thebibliography}
\end{document}